\documentclass[11pt]{article}

\usepackage[a4paper,margin=30mm]{geometry}
\usepackage[T1]{fontenc}
\usepackage{lmodern}
\usepackage{microtype}
\usepackage{amsmath,amssymb,amsthm,mathtools}
\usepackage[hidelinks]{hyperref}
\hypersetup{
  pdftitle={Metricity of separable quantum transport with the Hilbert--Schmidt cost},
  pdfauthor={Tomasz Miller}
}

\theoremstyle{plain}
\newtheorem{theorem}{Theorem}[section]
\newtheorem{lemma}[theorem]{Lemma}
\newtheorem{proposition}[theorem]{Proposition}
\newtheorem{corollary}[theorem]{Corollary}
\theoremstyle{remark}
\newtheorem{remark}[theorem]{Remark}

\newcommand{\RE}{\operatorname{Re}}
\newcommand{\IM}{\operatorname{Im}}

\newcommand{\ext}{\operatorname{ext}}
\newcommand{\Tr}{\operatorname{Tr}}

\newcommand{\Gammasep}{\Gamma_{\mathrm{sep}}}
\newcommand{\dsep}{d_{\mathrm{sep}}}
\newcommand{\cA}{\mathcal A}
\newcommand{\cC}{\mathcal C}
\newcommand{\cD}{\mathcal D}

\newcommand{\sC}{\mathbb C}
\newcommand{\sR}{\mathbb R}
\newcommand{\ip}[2]{\left\langle #1,#2\right\rangle}
\newcommand{\ket}[1]{\lvert #1\rangle}
\newcommand{\bra}[1]{\langle #1\rvert}
\newcommand{\proj}[1]{\ket{#1}\bra{#1}}
\newcommand{\Id}{I}

\title{Metricity of separable quantum optimal transport \\ with the Hilbert--Schmidt cost}
\author{Tomasz Miller\\[1mm]
\small Copernicus Center for Interdisciplinary Studies,\\[-1mm]
\small Jagiellonian University, Krak\'ow, Poland}
\date{\today}

\begin{document}

\maketitle

\begin{abstract}
We prove that the square root of the separable quantum optimal transport cost associated with the orthogonal projection onto the antisymmetric subspace defines a genuine distance between density matrices. Equivalently, this establishes the triangle inequality for the order-two Beatty--Fran\c ca quantum optimal transport construction induced by the Hilbert--Schmidt distance between pure states. The result also proves metricity of the corresponding distance derived from separable SWAP fidelity. The proof replaces the unavailable gluing argument by convex-roof duality and a dimension-independent interpolation result for Hermitian operators.

\end{abstract}



\section{Introduction}
\label{sec::intro}

Classical Monge--Kantorovich optimal transport starts from a metric on a space of points
and lifts it, by minimizing an average transportation cost over
couplings (or transport plans), to a metric on probability measures. The triangle
inequality is normally proved by gluing two optimal couplings along
their common marginal. This strategy, however, becomes problematic for quantum states. A
density matrix admits many inequivalent decompositions into pure
states, and two bipartite quantum couplings with a common marginal need
not be marginals of a single tripartite state themselves. This quantum marginal
obstruction is one of the principal reasons why metricity remains
delicate for static coupling-based versions of quantum optimal
transport \cite{BeattySurvey}.

Various coupling-based approaches to quantum optimal transport 
have been proposed (see, for example, the recent overview \cite{BeattySurvey}). 
One construction uses the orthogonal projection
onto the antisymmetric subspace as its cost operator (or some more general `quantum cost operator')
\cite{ColeEtAl,FriedlandEtAl}. However, this construction was later shown to violate the triangle inequality even on diagonal states \cite{Miller26}.

A different route, introduced and studied by Beatty and Stilck Fran\c ca~\cite{BeattyFranca}, 
begins with an actual metric on pure states and minimizes its $p$th power over
pure-product state decompositions of the two marginals\footnote{A special case of this approach was considered earlier by T\'oth and Pitrik \cite{TothPitrik23}.}. Equivalently, the
optimization is performed only over the \emph{separable} bipartite couplings.  This approach is
faithful, accommodates different pure-state metrics, and retains
a direct transportation interpretation. However, the abovementioned quantum marginal obstruction persists; since the two
optimal transports meeting at an intermediate density matrix can use
different pure-state decompositions of that matrix, the classical
gluing proof remains unavailable. Beatty and Stilck Fran\c ca
therefore left the triangle inequality as an open problem.

The present paper resolves that problem for the natural Hilbert--Schmidt pure-state metric. If $P_x=\proj{x}$ and $P_z=\proj{z}$ are
rank-one projections, set
\begin{align*}
 d_{\textup{HS}}(P_x,P_z)
 =\tfrac{1}{\sqrt{2}} \| P_x-P_z \|_2
 =\sqrt{1-|\langle x,z \rangle|^2},
\end{align*}
and consider the order-$2$ Beatty--Fran\c ca construction generated by $\tfrac{1}{\sqrt{2}}d_{\textup{HS}}$, namely
\begin{align}
 \dsep(\rho,\sigma) = \sqrt{\tfrac{1}{2} \min \sum\nolimits_j p_j d_{\textup{HS}}^2(P_{x_j},P_{z_j})},
 \label{dsepBSF}
\end{align}
where the minimum runs over all finite ensembles $(p_j,x_j,z_j)_j$ with $p_j > 0$ and unit vectors $x_j,z_j \in \sC^n$ satisfying $\sum\nolimits_j p_j P_{x_j} = \rho$ and $\sum\nolimits_j p_j P_{z_j} = \sigma$ (cf. \cite[Definitions 8 \& 9]{BeattyFranca}). In this paper we show that $\dsep$, defined below in a slightly different but equivalent way, does satisfy the triangle inequality and is indeed a genuine distance.

The same quantity also appears in the study of various types of separable fidelities \cite{TothPitrik26}. 
In particular, T\'oth and Pitrik proved the triangle inequality when the intermediate state is pure \cite[Theorem 10]{TothPitrik26}. Furthermore, Borsoni~\cite{Borsoni} identified the Beatty--Fran\c ca construction with a folded Kantorovich semidistance.
The associated chain construction always satisfies the triangle inequality, but whether the original one-step transport cost already
coincides with that chain metric remained the central subadditivity question. Our result shows that it does for the metric $d_{\textup{HS}}$, in
every finite dimension.

A similar but distinct metricity problem was recently solved by
Wirth~\cite{Wirth} for the De Palma--Trevisan channel-based quantum
Wasserstein divergence.  That construction optimizes over quantum
channels and uses a different quadratic cost.  The present theorem
concerns pure-product state decompositions, or equivalently separable
couplings, and neither result follows from the other.

The argument below can be loosely summarized as follows. We first
express the transport cost $T := \dsep^2$ as a convex roof and use standard
finite-dimensional convex duality to obtain an exact dual problem involving
two Hermitian operators satisfying a feasibility constraint. Given two such operators $K,M$, we construct 
an interpolating Hermitian operator $L$ satisfying two inequalities---one featuring $K$ and the other featuring $M$---which can then be used to establish the desired triangle inequality for any intermediate state. All this can be regarded as a ``dual substitute'' for the missing gluing
operation.

The paper is organized as follows. Section~\ref{sec::2} sets the stage and states the main result. Section~\ref{sec::3} derives the exact dual problem from finite-dimensional convex analysis. Section~\ref{sec::4} records the non-strict Finsler lemma in the form used here. Section~\ref{sec::5} proves Lemma~\ref{lem:interpolation} concerning the construction of the interpolating operator, and Section~\ref{sec::6} uses it to establish the triangle inequality. Finally, Section~\ref{sec::7} contains a brief outlook, translating the result into the language of separable SWAP fidelity and folded optimal transport.

\section{Setting and main result}
\label{sec::2}

Let $H_n$ denote the space of $n \times n$ Hermitian matrices. The set of quantum states (density matrices) is
\begin{align*}
 \cD_n := \{\rho \in H_n \, | \ \rho\succeq0,\ \Tr\rho=1\}.
\end{align*}
We use the convention that the first marginal of a bipartite state $\Omega$ is $\Tr_2\Omega$ and the second is $\Tr_1\Omega$. Thus the set of separable quantum couplings of $\rho, \sigma \in \cD_n$ is
\begin{align*}
 \Gammasep(\rho,\sigma) :=
 \left\{
 \Omega\in\cD_{n^2} \, | \
 \Omega\ \text{is separable},\
 \Tr_2\Omega=\rho,\
 \Tr_1\Omega=\sigma
 \right\}.
\end{align*}
This set is nonempty because it contains $\rho\otimes\sigma$. It is also compact, because the marginal constraints define a closed subset of the set of separable states, which is itself compact as the convex hull of the compact set of pure product states (in finite dimensions).

Let $S$ denote the SWAP operator, $S(x \otimes z) = z \otimes x$, and let $P_{\cA} := \frac{I-S}{2}$ be the orthogonal projection onto the antisymmetric subspace. In an orthonormal basis,
\begin{align*}
 P_{\cA} = \sum_{1\leq i<j\leq n} \ket{i\wedge j}\bra{i\wedge j}, \qquad \ket{i\wedge j} = \frac{\ket{ij}-\ket{ji}}{\sqrt2}.
\end{align*}
We consider the separable analogue of the optimal transport problem studied in \cite{ColeEtAl,FriedlandEtAl}, in which only the \emph{separable} quantum couplings are allowed. The optimal transport cost and its associated `separable quantum $2$-Wasserstein distance' are thus defined as
\begin{align}
 T(\rho,\sigma) := \min_{\Omega\in\Gammasep(\rho,\sigma)}\Tr(P_{\cA}\Omega),
 \qquad
 \dsep(\rho,\sigma):=\sqrt{T(\rho,\sigma)}.
 \label{eq:def-T-d}
\end{align}
Compactness of $\Gammasep(\rho,\sigma)$ means that the minimum is attained. 

For any unit vector $x \in \sC^n$ let $P_x:=\proj{x}$ denote the corresponding projection operator. It is convenient to introduce the cost function
\begin{align}
 c(x,z) := \Tr\left(P_\cA(P_x\otimes P_z)\right) = \tfrac{1}{2} (1-|\langle x,z \rangle|^2) = \tfrac{1}{2} d_{\textup{HS}}^2(P_x,P_z)
 \label{eq:pure-cost}
\end{align}
for any unit vectors $x,z\in\sC^n$.

As announced in Section \ref{sec::intro}, \eqref{eq:def-T-d} is the $p=2$ instance of the Beatty--Fran\c ca construction~\cite{BeattyFranca} based upon the pure-state distance $\tfrac{1}{\sqrt{2}}d_{\textup{HS}}$. Indeed, every separable coupling $\Omega \in \Gammasep(\rho,\sigma)$ can be refined, by spectrally decomposing its product factors, into a finite sum
\begin{align}
 \Omega = \sum\nolimits_j p_j P_{x_j} \otimes P_{z_j},
 \label{sepcoup}
\end{align}
and \eqref{eq:pure-cost} gives
\begin{align*}
 \Tr(P_\cA \Omega) = \sum\nolimits_j p_j c(x_j,z_j).
\end{align*}
Conversely, every ensemble $(p_j,x_j,z_j)_j$ satisfying the conditions accompanying \eqref{dsepBSF} defines, through \eqref{sepcoup}, a separable coupling with prescribed marginals. Hence the two minima agree and definition \eqref{eq:def-T-d} of $\dsep$ is indeed equivalent to \eqref{dsepBSF}.

Let us first show that $\dsep$ is a semidistance.
\begin{proposition}
\label{prop:elementary-axioms}
The function $\dsep$ is nonnegative and symmetric, and
\begin{align*}
 \dsep(\rho,\sigma) = 0 \quad\Longleftrightarrow\quad \rho=\sigma.
\end{align*}
\end{proposition}
\begin{proof}
Nonnegativity is immediate from $P_\cA\succeq0$.  If $\Omega$ is a separable coupling of $\rho$ and $\sigma$, then $S\Omega S$ is a separable coupling of $\sigma$ and $\rho$. Since $SP_\cA S=P_\cA$, the cost is unchanged, proving symmetry.

If $\rho=\sigma$, choose a pure-state decomposition
\begin{align*}
 \rho=\sum\nolimits_j p_j P_{x_j}
\end{align*}
and set $\Omega := \sum\nolimits_j p_j P_{x_j} \otimes P_{x_j}$. This is a separable coupling of $\rho$ with itself which satisfies $\Tr(P_{\cA}\Omega)=0$. Hence $\dsep(\rho,\rho)=0$.

Conversely, suppose $\dsep(\rho,\sigma)=0$ and refine an optimal separable coupling $\Omega$ into pure product form \eqref{sepcoup}. By \eqref{eq:pure-cost} and nonnegativity of all summands,
\begin{align*}
 0=\Tr(P_\cA\Omega)=\sum\nolimits_j p_j c(x_j,z_j)
\end{align*}
implies $|\langle x_j,z_j \rangle|=1$ whenever $p_j>0$.  Thus $P_{x_j}=P_{z_j}$ for every such $j$, and the two marginals of $\Omega$ coincide.  Therefore $\rho = \sigma$.
\end{proof}

Our main result is the following.
\begin{theorem}
\label{thm:main}
For every $n$, the function $\dsep$ defined in \eqref{eq:def-T-d} satisfies the triangle inequality.  Consequently, in view of Proposition~\ref{prop:elementary-axioms}, it is a metric on $\cD_n$.  Explicitly, for all $\rho,\tau,\sigma\in\cD_n$,
\begin{align*}
 \dsep(\rho,\sigma) \leq \dsep(\rho,\tau) + \dsep(\tau,\sigma).
\end{align*}
\end{theorem}

In order to prove it, we shall first establish the following lemma which replaces the standard gluing argument.
\begin{lemma}
\label{lem:interpolation}
Let $K,M \in H_n$ satisfy
\begin{align}
 \ip{x}{Kx}+\ip{z}{Mz}
 \leq
 c(x,z)
 \label{eq:product-feasibility}
\end{align}
for all unit $x,z\in\sC^n$. Let $\lambda>0$.  Then there exists $L \in H_n$ such that, for all unit $x,y,z$,
\begin{align}
 \ip{x}{Kx}+\ip{y}{Ly}
 &\leq (1+\lambda)c(x,y),
 \label{eq:interpolating-1}\\
 -\ip{y}{Ly}+\ip{z}{Mz}
 &\leq (1+\lambda^{-1}) c(y,z).
 \label{eq:interpolating-2}
\end{align}
\end{lemma}

\section{Convex-roof and dual formulations}
\label{sec::3}

\begin{proposition}
\label{prop:duality}
For all $\rho,\sigma\in\cD_n$,
\begin{align}
 T(\rho,\sigma)
 =
 \sup \left\{ \Tr(K\rho)+\Tr(M\sigma) \, \left| \,
 \begin{array}{l}
 K,M\in H_n,\\
 \ip{x}{Kx}+\ip{z}{Mz} \leq c(x,z) \ \textup{for all unit }x,z \in \sC^n
 \end{array}
 \right.
 \right\}.
 \label{eq:duality}
\end{align}
\end{proposition}

\begin{proof}
Let $\cC := \cD_n\times\cD_n$, regarded as a compact convex subset of the $\sR$-linear space $H_n \oplus H_n$. Its extreme points are exactly the pairs $(P_x,P_z)$ of rank-one projections. In particular, $\ext \cC$ is compact. Define the continuous function $g$ on $\operatorname{ext}\cC$ by
\begin{align*}
 g(P_x,P_z)=c(x,z).
\end{align*}
Since $H_n$ is a real vector space of dimension $n^2$, we have
\begin{align*}
 \operatorname{aff}\cC
 =
 \left\{
 (A,B)\in H_n\oplus H_n\,|\,\Tr A=\Tr B=1
 \right\},
 \qquad
 \dim(\operatorname{aff}\cC)=2(n^2-1).
\end{align*}
Extend $g$ to an extended-real-valued function on $\operatorname{aff}\cC$ by
\begin{align*}
 \widetilde g(a)
 :=
 \begin{cases}
  g(a),&a\in\ext\cC,\\
  +\infty,&a\notin\ext\cC.
 \end{cases}
\end{align*}
A standard application of Carath\'{e}odory's theorem~\cite[Corollary~17.1.5]{Rockafellar} shows that the greatest convex minorant of $\widetilde g$, denoted $\operatorname{conv}\widetilde g$, is obtained by taking the infimum of $\sum\nolimits_j p_j g(a_j)$ over all finite extreme-point convex decompositions $a = \sum\nolimits_j p_j a_j$, where at most $2n^2-1$ points are needed. We claim that this infimum is attained in the present setting. Indeed, consider the set of convex decompositions (with zero coefficients allowed)
\begin{align*}
\mathfrak{A} := \Big\{ (p_j,a_j)_{j=1}^{2n^2-1} \, \Big| \ \forall j \ p_j \geq 0, a_j \in \ext\cC, \ \sum\nolimits_j p_j = 1 \Big\}.
\end{align*}
This set is compact, being homeomorphic to the product of the probability $2(n^2-1)$-simplex and $(\ext\cC)^{2n^2-1}$. Additionally, for any fixed $a\in\cC$ introduce
\begin{align*}
\mathfrak{A}_a := \Big\{ (p_j,a_j)_j \in \mathfrak{A} \, \Big| \ \sum\nolimits_j p_j a_j = a \Big\}.
\end{align*} 
Clearly, $\mathfrak{A}_a$ is a closed---and hence compact---subset of $\mathfrak{A}$. Since $\cC=\operatorname{conv}(\ext\cC)$, the set $\mathfrak{A}_a$ is also nonempty. Furthermore, the continuity of the map $\mathfrak{A} \ni (p_j,a_j)_j \mapsto \sum\nolimits_j p_j g(a_j)$ means that it attains its infimum on $\mathfrak{A}_a$ for any chosen $a\in\cC$.

As a result, we obtain the convex roof of $g$ as the restricted map $g^{\cup} := (\operatorname{conv}\widetilde g)|_{\cC}$, i.e.,
\begin{align}
 g^{\cup}(\rho,\sigma)
 :=
 \min\Big\{
 \sum\nolimits_j p_j c(x_j,z_j) \, \Big| \ \forall j \; p_j \geq 0, \
 \sum\nolimits_j p_j P_{x_j} = \rho, \ \sum\nolimits_j p_j P_{z_j} = \sigma
 \Big\},
 \label{convexroof}
\end{align}
in agreement with the standard convex-roof formula~\cite[equation~(25)]{Uhlmann}. Notice that $g^{\cup}= \dsep^2 = T$ by formula \eqref{dsepBSF} proven in the previous section.

Thus, in order to establish identity \eqref{eq:duality}, it remains to derive the dual formula for the convex roof $g^{\cup}$. To this end we shall use the Fenchel--Moreau theorem \cite{Bauschke}, which says that a proper function on a (Hausdorff) locally convex space is equal to its biconjugate iff it is convex and lower semicontinuous. Since we already know that $g^{\cup}$ is proper and convex, we only need to prove its lower semicontinuity.

Let us therefore take any $a_k\to a$ in $\cC$, and pass to a subsequence along which $g^{\cup}(a_k)$ converges to $\liminf_k g^{\cup}(a_k)$. For every $k$, choose a minimizing decomposition of $a_k$ with $2n^2-1$ terms, allowing zero coefficients. By compactness of $\mathfrak{A}$, after passing to a further subsequence, all coefficients and all extreme points in these decompositions converge. Their limits give an extreme-point decomposition $a=\sum\nolimits_j p_j e_j$, and continuity of $g$ together with convexity of $g^{\cup}$ allow us to write
\begin{align*}
 \liminf_{k\to\infty}g^{\cup}(a_k)
 =\sum\nolimits_j p_j g(e_j)=\sum\nolimits_j p_j g^{\cup}(e_j)
 \geq g^{\cup}(a).
\end{align*}
This proves the lower semicontinuity of $g^{\cup}$.

On the finite-dimensional real vector space $H_n\oplus H_n$, define $f:H_n\oplus H_n\to(-\infty,+\infty]$ by
\begin{align*}
 f(a)
 :=
 \begin{cases}
  g^{\cup}(a),&a\in\cC,\\
  +\infty,&a\notin\cC.
 \end{cases}
\end{align*}
Since $\cC$ is closed and $g^{\cup}$ is lower semicontinuous on $\cC$, the function $f$ is proper, lower semicontinuous, and convex. By the Fenchel--Moreau theorem, $f$ equals its biconjugate and can therefore be represented as the pointwise supremum of its affine minorants. On $\cC$, the condition that an affine function on $H_n\oplus H_n$ minorize $f$ is equivalent to its restriction to $\cC$ minorizing $g^{\cup}$. Conversely, every affine function on $\cC$ extends to one on $H_n\oplus H_n$. It follows that, for any $a\in\cC$,
\begin{align*}
 g^{\cup}(a)
 =
 \sup\left\{
 \ell(a)\,|\,
 \ell\text{ is affine on }\cC,\ \ell\leq g^{\cup}\text{ on }\cC
 \right\}.
\end{align*}

It remains to express the minorant condition solely on the extreme points. If $\ell\leq g^{\cup}$ on $\cC$, then $\ell(e)\leq g(e)$ for every $e\in\ext\cC$. Conversely, if the latter inequalities hold, then for every extreme-point decomposition $a=\sum\nolimits_j p_j e_j$,
\begin{align*}
 \ell(a)=\sum\nolimits_j p_j \ell(e_j) \leq \sum\nolimits_j p_j g(e_j).
\end{align*}
Taking the minimum over all such decompositions gives $\ell(a)\leq g^{\cup}(a)$. We have therefore proved
\begin{align}
 g^{\cup}(a)
 =
 \sup\left\{
 \ell(a) \, | \,
 \ell\text{ is affine on }\cC, \
 \ell(e)\leq g(e) \ \text{for every }e \in \ext\cC
 \right\}.
 \label{eq:affine-dual}
\end{align}

Every real-valued affine function on $\cC$ can be written as
\begin{align*}
 \ell(\rho,\sigma)
 =r+\Tr(K\rho)+\Tr(M\sigma)
\end{align*}
for some $K,M \in H_n$ and $r\in\sR$. But since $\Tr\rho=\Tr\sigma=1$, the constant $r$ can be absorbed into one of the Hermitian matrices and we can thus assume that
\begin{align*}
 \ell(\rho,\sigma)=\Tr(K\rho)+\Tr(M\sigma).
\end{align*}
The extreme-point constraint in \eqref{eq:affine-dual} is then exactly
\begin{align*}
 \ip{x}{Kx}+\ip{z}{Mz} \leq c(x,z)
\end{align*}
for all unit $x,z \in \sC^n$. Substitution into \eqref{eq:affine-dual} gives \eqref{eq:duality}, concluding the proof.
\end{proof}

\begin{remark}
The dual formula \eqref{eq:duality} may equivalently be written as
\begin{align*}
 T(\rho,\sigma)
 =
 \sup \left\{ \Tr(K\rho)+\Tr(M\sigma) \, \left| \,
 \begin{array}{l}
 K,M\in H_n,\\
 P_\cA-K\otimes I-I\otimes M \textnormal{ is block-positive}
 \end{array}
 \right.
 \right\},
\end{align*}
where we recall that a bipartite Hermitian operator $W$ is block-positive if it satisfies $\ip{x\otimes z}{W(x\otimes z)}\geq0$ for all $x,z\in\sC^n$. 

It is instructive to compare this formula with the dual formulation of the minimization of $\Tr(P_\cA\Omega)$ over \emph{all} quantum couplings $\Omega$ rather than only the separable ones, given in \cite[Theorem 3.2]{ColeEtAl}. The sole formal difference is that the latter dual problem imposes the stronger constraint $P_\cA-K\otimes I-I\otimes M\succeq 0$. Thus, restricting the primal problem to separable couplings corresponds in the dual problem to relaxing positive semidefiniteness to block positivity.
\end{remark}

\section{The standard non-strict Finsler lemma}
\label{sec::4}

The next elementary form of Finsler's lemma will produce the scalar shift needed to construct the interpolating operator $L$ in Lemma \ref{lem:interpolation}.

\begin{lemma}[Non-strict Finsler lemma]
\label{lem:finsler}
Let $H,J$ be Hermitian operators on a finite-dimensional complex Hilbert space. Suppose that $J$ is indefinite, meaning that its quadratic form takes both positive and negative values. If
\begin{align}
 \ip{w}{Hw}\geq0
 \qquad\text{whenever}\qquad
 \ip{w}{Jw}=0,
 \label{eq:finsler-hyp}
\end{align}
then there exists $h\in\sR$ such that
\begin{align}
 H-hJ \succeq 0.
 \label{eq:finsler-concl}
\end{align}
\end{lemma}

This is the indefinite case of the standard non-strict Finsler lemma; see, for example, Lemma~2 of Meijer et al.~\cite{MeijerEtAl}. The result in that reference is stated for real symmetric matrices. The Hermitian formulation above follows immediately by realification: for a Hermitian matrix $C$, set
\begin{align}
 \mathfrak{R}(C)
 :=
 \begin{bmatrix}
  \RE C & -\IM C\\
  \IM C & \RE C
 \end{bmatrix},
 \label{eq:realification}
\end{align}
where Re and Im are taken \emph{entrywise}. If $w=u+iv$, then
\begin{align*}
 \ip{w}{Cw}
 =
 \begin{bmatrix}u\\v\end{bmatrix}^\top
 \mathfrak{R}(C)
 \begin{bmatrix}u\\v\end{bmatrix}.
\end{align*}
This idedntity shows that realification preserves indefiniteness and Loewner order. Moreover, $\mathfrak{R}(H-hJ)=\mathfrak{R}(H)-h\mathfrak{R}(J)$, so the cited real result yields the asserted Hermitian version.

\section{Proof of Lemma~\ref{lem:interpolation}}
\label{sec::5}

\begin{proof}[Proof of Lemma~\ref{lem:interpolation}]
For a Hermitian operator $A \in H_n$ and a unit vector $x \in \sC^n$, we shall abbreviate $A_x:=\ip{x}{Ax}$.

Fix $K,M \in H_n$ satisfying \eqref{eq:product-feasibility}. We construct the interpolating operator $L$ in four steps.

\medskip
\noindent\textbf{Step 1.}
On $\sC^n\oplus\sC^n$, define
\begin{align}
 H :=
 \begin{bmatrix}
 \tfrac12 I-M & -\tfrac12 I \\
 -\tfrac12 I & \tfrac12 I-K
 \end{bmatrix},
 \qquad
 J :=
 \begin{bmatrix}
 I & 0\\
 0 &-I
 \end{bmatrix}.
 \label{eq:M-J-direct}
\end{align}
The operator $J$ is indefinite. We verify the hypothesis of Lemma~\ref{lem:finsler}. If $w=(u,v)$ satisfies
\begin{align}
 \ip{w}{Jw } = \|u\|^2-\|v\|^2=0,
 \label{eq:direct-null-cone}
\end{align}
then either $w=0$, or $r:=\|u\|=\|v\|>0$.  In the latter case put
\begin{align}
 z=\frac ur, \qquad x=\frac vr.
 \label{eq:direct-unit-vectors}
\end{align}
These are unit vectors. Using \eqref{eq:product-feasibility} we obtain
\begin{align*}
 \tfrac1{r^2}\ip{w}{Hw}
 &=1-K_x-M_z-\RE\ip{x}{z}\\
 &\geq
 1-\RE\ip{x}{z}-\tfrac12\left(1-|\langle x, z \rangle|^2\right)\\
 &=\tfrac12|1-\langle x, z \rangle|^2\geq0.
\end{align*}
Thus $\ip{w}{Hw}\geq0$ throughout the null cone of $J$. The non-strict Finsler lemma supplies $h\in\sR$ such that $H-hJ\succeq 0$. Expanding this inequality gives
\begin{align}
 \begin{bmatrix}
 \big(\tfrac12 - h\big) I - M & -\tfrac12 I \\[1mm]
 -\tfrac12 I & \big(\tfrac12 + h\big) I - K
 \end{bmatrix}
 \succeq 0.
 \label{eq:central-block}
\end{align}

\medskip
\noindent\textbf{Step 2.}
For convenience, denote the diagonal blocks in \eqref{eq:central-block} by
\begin{align}
 A := \big(\tfrac12 - h\big) I - M,
 \qquad
 B := \big(\tfrac12 + h\big) I - K.
 \label{eq:A-B}
\end{align}
Positive-semidefiniteness \eqref{eq:central-block} immediately implies that $A,B \succeq 0$, but actually both these blocks are positive \emph{definite}. Indeed, suppose $Au=0$. Then for $w = (u,0)$ we have $(H-hJ)w = (0, -\tfrac{1}{2}u)$, and so $\langle w, (H-hJ)w \rangle = 0$. But since $H-hJ\succeq 0$, the vanishing quadratic form implies that $(H-hJ)w = 0$, and altogether we thus obtain $u = 0$. An analogous argument proves $B\succ0$. 

The Schur complement of $A$ is thus well-defined and its positive-semidefiniteness gives
\begin{align}
 B \succeq \tfrac14 A^{-1}.
 \label{eq:B-schur}
\end{align}

Now define, for any fixed $\lambda > 0$,
\begin{align}
 X := \tfrac{1+\lambda^{-1}}2 I - h I - M = A + \tfrac1{2\lambda} I,
 \qquad
 Y := \tfrac{1+\lambda}2 I + h I - K = B + \tfrac\lambda2 I.
 \label{eq:X-Y}
\end{align}
Notice that, in particular, $X,Y \succ 0$.  We shall later need the following estimate.
\begin{align}
 \lambda^{-1}X^{-1} + \lambda Y^{-1} \preceq 2 I.
 \label{eq:inverse-order}
\end{align}
Indeed, by \eqref{eq:B-schur} and the order-reversing property of inversion,
\begin{align*}
 \lambda^{-1}X^{-1} + \lambda Y^{-1} 
 \preceq \lambda^{-1}\left(A+\tfrac1{2\lambda}I\right)^{-1} + \lambda\left(\tfrac14 A^{-1} +\tfrac\lambda2 I \right)^{-1}.
\end{align*}
For every scalar $a>0$ we have that
\begin{align}
 \frac{\lambda^{-1}}{a+\tfrac1{2\lambda}} + \frac\lambda{\tfrac1{4a}+\tfrac\lambda2}
 =
 \frac{2}{1+2\lambda a}+\frac{4\lambda a}{1+2\lambda a} = 2.
 \label{eq:resolvent-scalar}
\end{align}
Functional calculus for the positive operator $A$ therefore yields \eqref{eq:inverse-order}.

We shall also need another estimate, which comes directly from the original feasibility inequality. For unit $u,v$, \eqref{eq:X-Y} and \eqref{eq:product-feasibility} give
\begin{align}
 Y_u+X_v
 &=\frac{(1+\lambda)^2}{2\lambda}-K_u-M_v \geq
 \frac{\lambda^2+\lambda+1}{2\lambda} + \frac12|\langle u, v\rangle|^2.
 \label{eq:XY-product}
\end{align}

\medskip
\noindent\textbf{Step 3.}
Define another auxiliary Hermitian operator
\begin{equation}
 Q:=\frac{\lambda^2-1}{2\lambda}I + \frac{1+\lambda}{4\lambda^2}X^{-1} - \frac{\lambda(1+\lambda)}4 Y^{-1}.
 \label{eq:Q}
\end{equation}
We claim that, for all unit $u,v \in \sC^n$,
\begin{align}
 X_v+Q_u & \geq \frac{1+\lambda}{2\lambda}|\langle u, v\rangle|^2,
 \label{eq:interpolation-concl-1}\\
 Y_u-Q_v & \geq \frac{1+\lambda}{2}|\langle u, v\rangle|^2.
 \label{eq:interpolation-concl-2}
\end{align}
It suffices to prove \eqref{eq:interpolation-concl-1}. Indeed, simultaneously interchanging
\begin{align*}
 \lambda\longmapsto\lambda^{-1}, \qquad X\longleftrightarrow Y, \qquad u\longleftrightarrow v
\end{align*}
preserves \eqref{eq:inverse-order} and \eqref{eq:XY-product}, changes $Q$ into $-Q$, and transforms \eqref{eq:interpolation-concl-1} into \eqref{eq:interpolation-concl-2}.

In order to prove \eqref{eq:interpolation-concl-1}, fix unit $u,v$ and set
\begin{align*}
s:=X_v, \qquad t:=|\langle u, v\rangle|^2.
\end{align*}
Notice that the Cauchy--Schwarz inequality gives
\begin{align}
t = |\langle X^{-1/2}u, X^{1/2}v \rangle|^2 \leq (X^{-1})_u X_v = s(X^{-1})_u.
\label{eq:CS}
\end{align}
Taking expectations of \eqref{eq:inverse-order} and \eqref{eq:Q} with respect to $u$, we can write that
\begin{align*}
 X_v + Q_u - \frac{1+\lambda}{2\lambda}t & = s + \frac{\lambda^2 - 1}{2\lambda} + \frac{1+\lambda}{4}\left( \lambda^{-2}(X^{-1})_u - \lambda (Y^{-1})_u \right) - \frac{1+\lambda}{2\lambda}t
 \\
 & \geq s-\frac{1+\lambda}{2\lambda}(1+t) + \frac{(1+\lambda)^2}{4\lambda^2}\frac{t}{s}
 \\
 & = \frac{\big(2\lambda s-(1+\lambda)\big) \big(2\lambda s-(1+\lambda)t\big)}{4\lambda^2s}
\end{align*}
where in the inequality we have also used \eqref{eq:CS} to bound $(X^{-1})_u$ from below. Observe that the rightmost expression is nonnegative if $s\leq\tfrac{1+\lambda}{2\lambda}t$ or $s\geq\frac{1+\lambda}{2\lambda}$. Thus, from now on we can assume that
\begin{align}
 \frac{1+\lambda}{2\lambda}t < s < \frac{1+\lambda}{2\lambda}.
 \label{eq:s-range}
\end{align}
In this case, we need a more careful approach.

First, notice that \eqref{eq:inverse-order} implies $X^{-1} \preceq 2\lambda I$, and hence
\begin{align}
 s\geq\frac1{2\lambda}.
 \label{eq:s-lower}
\end{align}
Second, let us fix $v$ and interpret \eqref{eq:XY-product} as the operator inequality
\begin{align*}
 Y \succeq \left(\frac{1+\lambda+\lambda^2}{2\lambda}-s\right) I + \frac12 P_v = \left(\frac{1+\lambda+\lambda^2}{2\lambda}-s\right)(I - P_v) + \left(\frac{(1+\lambda)^2}{2\lambda}-s\right) P_v.
\end{align*}
Its right-hand side is positive definite by \eqref{eq:s-range}. Inverting both sides of this inequality and taking its expectation with respect to $u$, we obtain
\begin{align}
 (Y^{-1})_u \leq \frac{1-t}{\frac{1+\lambda+\lambda^2}{2\lambda}-s} + \frac{t}{\frac{(1 + \lambda)^2}{2\lambda}-s}.
 \label{eq:Y-1}
\end{align}
This time, taking the expectation of \eqref{eq:Q} with respect to $u$ allows us to write, with the help of \eqref{eq:CS} and \eqref{eq:Y-1},
\begin{align*}
& X_v+Q_u-\frac{1+\lambda}{2\lambda}t = s + \frac{\lambda^2 - 1}{2\lambda} + \frac{1+\lambda}{4}\left( \lambda^{-2}(X^{-1})_u - \lambda (Y^{-1})_u \right) - \frac{1+\lambda}{2\lambda}t
 \\
 & \geq s + \frac{\lambda^2 - 1}{2\lambda} + \frac{(1+\lambda)t}{4\lambda^2 s} - \frac{(1+\lambda)\lambda}{4}\left( \frac{1-t}{\frac{\lambda^2+\lambda+1}{2\lambda}-s} + \frac{t}{\frac{(\lambda + 1)^2}{2\lambda}-s} \right) - \frac{1+\lambda}{2\lambda}t
\\
& = \left(1-\frac{(1+\lambda)t}{2\lambda s}\right)
 \frac{\left(s-\frac1{2\lambda}\right)
 \left(\frac{1+\lambda}{2\lambda}-s\right)}
 {\frac{1+\lambda+\lambda^2}{2\lambda}-s} 
 +
 \frac{(1+\lambda)t}{2\lambda s}
 \frac{\frac\lambda2
 \left(\frac{1+\lambda}{2\lambda}-s\right)^2}
 {\left(\frac{(1+\lambda)^2}{2\lambda}-s\right)
 \left(\frac{1+\lambda+\lambda^2}{2\lambda}-s\right)}.
\end{align*}
The expression on the right-hand side is affine in $\tfrac{(1+\lambda)t}{2\lambda s}$, which lies in $[0,1)$ by \eqref{eq:s-range}. The right-hand side is therefore a convex combination of two nonnegative quantities: all denominators are positive by \eqref{eq:s-range}, and the numerators are nonnegative by \eqref{eq:s-range} and \eqref{eq:s-lower}. This proves \eqref{eq:interpolation-concl-1}, and the symmetry argument proves \eqref{eq:interpolation-concl-2}.

\medskip
\noindent\textbf{Step 4.}
We are finally ready to construct the Hermitian operator $L$ asserted in the lemma's statement. Namely, define
\begin{align}
 L := Q-h\Id
 =
 \left(\frac{\lambda^2-1}{2\lambda}-h\right)\Id
 +\frac{1+\lambda}{4\lambda^2}X^{-1}
 -\frac{\lambda(1+\lambda)}4Y^{-1}.
 \label{eq:explicit-interpolating-operator}
\end{align}
From \eqref{eq:X-Y} we have that
\begin{align*}
 K=\left(\frac{1+\lambda}{2}+h\right)I - Y,
 \qquad
 M=\left(\frac{1+\lambda^{-1}}2-h\right)I - X.
\end{align*}
Thus \eqref{eq:interpolation-concl-2} gives
\begin{align*}
 K_x+L_y &= \frac{1+\lambda}{2} - Y_x + Q_y \leq \frac{1+\lambda}{2} \left(1-|\langle x,y \rangle|^2\right),
\end{align*}
while \eqref{eq:interpolation-concl-1} gives
\begin{align*}
 -L_y+M_z &= \frac{1+\lambda^{-1}}2 - X_z - Q_y \leq \frac{1+\lambda^{-1}}2 \left(1-|\langle y,z \rangle|^2\right).
\end{align*}
These are exactly the desired inequalities \eqref{eq:interpolating-1} and \eqref{eq:interpolating-2}.
\end{proof}

\section{Proof of the triangle inequality}
\label{sec::6}

\begin{proof}[Proof of Theorem~\ref{thm:main}]
Fix $\rho,\tau,\sigma\in\cD_n$. If $d_{\mathrm{sep}}(\rho,\tau) = 0$, then Proposition~\ref{prop:elementary-axioms} gives $\rho=\tau$
and the asserted triangle inequality is trivially true. Similarly, if $d_{\mathrm{sep}}(\tau,\sigma) = 0$, then $\tau=\sigma$ and the assertion is again trivially true. We may therefore assume throughout the rest of the proof that $d_{\mathrm{sep}}(\rho,\tau), d_{\mathrm{sep}}(\tau,\sigma) > 0$.

With the above in mind, put $\lambda := d_{\mathrm{sep}}(\tau,\sigma)/d_{\mathrm{sep}}(\rho,\tau)$ and let $K,M\in H_n$ be arbitrary operators satisfying \eqref{eq:product-feasibility}. Lemma~\ref{lem:interpolation} supplies an interpolating operator $L\in H_n$ satisfying
\eqref{eq:interpolating-1}--\eqref{eq:interpolating-2}.

Let $\Omega_{12}\in\Gammasep(\rho,\tau)$ and $\Omega_{23}\in\Gammasep(\tau,\sigma)$ be optimal separable couplings whose pure-product decompositions read
\begin{align*}
 \Omega_{12} = \sum\nolimits_i p_iP_{x_i}\otimes P_{y_i}, 
 \qquad
 \Omega_{23} = \sum\nolimits_j q_jP_{y'_j}\otimes P_{z_j}.
\end{align*}
The two decompositions of the middle state $\tau$ need not coincide; this is exactly why Lemma~\ref{lem:interpolation} supplies a single interpolating operator $L$ rather than a common pure-state ensemble.

Multiplying \eqref{eq:interpolating-1} by $p_i$ and summing gives
\begin{align}
 \Tr(K\rho)+\Tr(L\tau)
 &\leq
 \frac{1+\lambda}{2}
 \sum\nolimits_i p_i\left(1-|\ip{x_i}{y_i}|^2\right)\notag\\
 &=(1+\lambda)T(\rho,\tau)
 =(1+\lambda)d^2_{\mathrm{sep}}(\rho,\tau).
 \label{eq:integrated-interpolation-1}
\end{align}
Likewise, summing \eqref{eq:interpolating-2} over the second separable coupling gives
\begin{equation}
 -\Tr(L\tau)+\Tr(M\sigma)
 \leq
 (1+\lambda^{-1})T(\tau,\sigma)
 =(1+\lambda^{-1})d^2_{\mathrm{sep}}(\tau,\sigma).
 \label{eq:integrated-interpolation-2}
\end{equation}
Adding \eqref{eq:integrated-interpolation-1} and \eqref{eq:integrated-interpolation-2} yields:
\begin{align*}
 \Tr(K\rho)+\Tr(M\sigma)
 &\leq
 (1+\lambda)d^2_{\mathrm{sep}}(\rho,\tau)+(1+\lambda^{-1})d^2_{\mathrm{sep}}(\tau,\sigma) = \left(d_{\mathrm{sep}}(\rho,\tau)+d_{\mathrm{sep}}(\tau,\sigma)\right)^2.
\end{align*}
Taking the supremum over all Hermitian $K,M$ satisfying \eqref{eq:product-feasibility} and using Proposition~\ref{prop:duality} yields
\begin{align*}
 T(\rho,\sigma)\leq\left(d_{\mathrm{sep}}(\rho,\tau)+d_{\mathrm{sep}}(\tau,\sigma)\right)^2.
\end{align*}
Taking square roots proves the desired triangle inequality.
\end{proof}

\section{Equivalent formulations and outlook}
\label{sec::7}

Drawing from \cite[Eq. (54)]{TothPitrik26}, define the separable SWAP fidelity by
\begin{align*}
 F_{S,\mathrm{sep}}(\rho,\sigma) := \max_{\Omega\in\Gammasep(\rho,\sigma)}\Tr(S\Omega).
\end{align*}
Since $P_\cA = (I-S)/2$, the optimal transport cost \eqref{eq:def-T-d} can be expressed in terms of $F_{S,\mathrm{sep}}$ via
\begin{align*}
 T(\rho,\sigma) = \tfrac12\left(1-F_{S,\mathrm{sep}}(\rho,\sigma)\right).
\end{align*}
We therefore obtain the following equivalent formulation of our main metricity result.
\begin{corollary}
For every $n$, $\sqrt{1-F_{S,\mathrm{sep}}}$ is a metric on $\cD_n$.
\end{corollary}

T\'oth and Pitrik also relate this metric to two separable-state modifications of De Palma--Trevisan-type quantum Wasserstein distances denoted $\tilde{D}_{\mathrm{DPT},\mathrm{decomp}}$ and $\tilde{D}_{\mathrm{DPT},\mathrm{sep}}$ \cite[Eqs.~(83) and (90)]{TothPitrik26}. In the special case where these distances are built on a full set of operators (cf. \cite[Section III C.2]{TothPitrik26}), it is proven that $\tilde{D}^2_{\mathrm{DPT},\mathrm{decomp}} = \tilde{D}^2_{\mathrm{DPT},\mathrm{sep}} = 2(1-F_{S,\mathrm{sep}})$. Thus, the above corollary certifies that both $\tilde{D}_{\mathrm{DPT},\mathrm{decomp}}$ and $\tilde{D}_{\mathrm{DPT},\mathrm{sep}}$ are genuine metrics in this case.

There is also a direct consequence for folded optimal transport. Borsoni's folded Kantorovich semidistance $\overline D_2$ associated with $\tfrac{1}{\sqrt{2}}d_{\textup{HS}}$ is precisely the Beatty--Fran\c ca quantity $W_2^{d_{\textup{HS}}/\sqrt{2}}$, whereas the folded Wasserstein distance $D_2$ is its chain envelope~\cite{Borsoni}.  The latter always satisfies $D_2\leq\overline D_2$. Conversely, Theorem~\ref{thm:main} implies that every finite chain $\rho_0,\ldots,\rho_m$ satisfies
\begin{align*}
 \overline D_2(\rho_0,\rho_m)
 \leq\sum_{j=1}^m\overline D_2(\rho_{j-1},\rho_j).
\end{align*}
Taking the infimum over chains gives the reverse inequality.  Hence
\begin{align*}
 D_2 = \overline D_2 = W_2^{d_{\textup{HS}}/\sqrt{2}} = \dsep.
\end{align*}
Thus, for the Hilbert--Schmidt pure-state geometry, passing to the chain envelope does not decrease the one-step separable transport cost.

The above reasoning seems to be specific to the quadratic cost in an essential way: the proof of Lemma~\ref{lem:interpolation} uses a Hermitian block matrix, Schur complements, and resolvents of positive operators. It does not by itself settle the corresponding metricity question for other orders $p$ or for general metrics on projective space.  Nevertheless, Lemma~\ref{lem:interpolation} suggests a broader strategy for such problems: replace the gluing of couplings with an interpolation of Hermitian operators. Whether an analogous result holds beyond the quadratic setting or beyond the underlying Hilbert--Schmidt pure-state distance remains open.

\section*{Acknowledgements}

The author acknowledges the assistance of \emph{ChatGPT 5.6 Sol} in the research process. All proofs, citations, and wording were checked, edited, and validated by the author.


\begin{thebibliography}{99}

\bibitem{Bauschke}
H.~H. Bauschke and P.~L. Combettes,
\emph{Convex analysis and monotone operator theory in Hilbert spaces},
CMS Books in Mathematics,
Springer, New York, 2011.
\href{https://doi.org/10.1007/978-1-4419-9467-7}
{doi:10.1007/978-1-4419-9467-7}.

\bibitem{BeattySurvey}
E.~Beatty,
\emph{Wasserstein distances on quantum structures: an overview},
arXiv:2506.09794 (2025).

\bibitem{BeattyFranca}
E.~Beatty and D.~Stilck Fran\c ca,
\emph{Order $p$ quantum Wasserstein distances from couplings},
Ann. Henri Poincar\'e \textbf{27} (2026), 787--845.
\href{https://doi.org/10.1007/s00023-025-01557-z}
{doi:10.1007/s00023-025-01557-z}.

\bibitem{Borsoni}
T.~Borsoni,
\emph{Metric convex extensions and optimal transport on classical and quantum structures}, arXiv:2512.01722v4 (2026).

\bibitem{ColeEtAl}
S.~Cole, M.~Eckstein, S.~Friedland, and K.~\.{Z}yczkowski,
\emph{On quantum optimal transport},
Math. Phys. Anal. Geom. \textbf{26} (2023), Article 14.
\href{https://doi.org/10.1007/s11040-023-09456-7}
{doi:10.1007/s11040-023-09456-7}.

\bibitem{FriedlandEtAl}
S.~Friedland, M.~Eckstein, S.~Cole, and K.~\.{Z}yczkowski,
\emph{Quantum Monge--Kantorovich problem and transport distance
between density matrices},
Phys. Rev. Lett. \textbf{129} (2022), 110402.
\href{https://doi.org/10.1103/PhysRevLett.129.110402}
{doi:10.1103/PhysRevLett.129.110402}.

\bibitem{MeijerEtAl}
T.~J. Meijer, K.~J.~A. Scheres, S.~van den Eijnden, T.~Holicki,
C.~W. Scherer, and W.~P.~M.~H. Heemels,
\emph{A unified non-strict Finsler lemma},
IEEE Control Syst. Lett. \textbf{8} (2024), 1955--1960.
\href{https://doi.org/10.1109/LCSYS.2024.3415473}
{doi:10.1109/LCSYS.2024.3415473}.

\bibitem{Miller26}
T.~Miller,
\emph{Comment on `Quantum Monge--Kantorovich Problem and Transport Distance between Density Matrices'}, arXiv:2607.07764v1 (2026).

\bibitem{Rockafellar}
R.~T. Rockafellar,
\emph{Convex analysis},
Princeton Mathematical Series, No.~28,
Princeton University Press, Princeton, NJ, 1970.

\bibitem{TothPitrik26}
G.~T\'oth and J.~Pitrik,
\emph{Quantum Wasserstein distance and its relation to several types
of fidelities}, arXiv:2506.14523v6 (2026).

\bibitem{TothPitrik23}
G.~T\'oth and J.~Pitrik,
\emph{Quantum Wasserstein distance based on an optimization over separable states}, Quantum \textbf{7} (2023), 1143.
\href{https://doi.org/10.22331/q-2023-10-16-1143}
{doi:10.22331/q-2023-10-16-1143}.

\bibitem{Uhlmann}
A.~Uhlmann,
\emph{Roofs and convexity},
Entropy \textbf{12} (2010), 1799--1832.
\href{https://doi.org/10.3390/e12071799}
{doi:10.3390/e12071799}.

\bibitem{Wirth}
M.~Wirth,
\emph{Triangle inequality for a quantum Wasserstein divergence},
arXiv:2511.20450 (2025).

\end{thebibliography}
\end{document}